\documentclass[a4paper,10pt,twoside]{article}

\usepackage{pst-fill,pst-grad}
\usepackage{textcomp}
\usepackage[english]{babel}
 \usepackage[utf8]{inputenc}
\usepackage{graphicx}
\usepackage{amsmath}
\usepackage{float}
\usepackage[matrix,arrow,curve]{xy}
\usepackage{pstricks} 
\usepackage{amsmath,amsfonts,verbatim,,afterpage,euscript,mathrsfs,amssymb}
\usepackage{amsfonts}
\usepackage{array}
\usepackage{amsmath,amssymb,amsthm}
\usepackage{psfrag}
\usepackage{dsfont}
\usepackage{makeidx}
\usepackage[nottoc, notlof, notlot]{tocbibind}
\makeindex
\usepackage{algorithmicx}
\usepackage{algpseudocode}
\usepackage{algorithm}
\usepackage{titlesec}

\usepackage{hyperref}
\hypersetup{colorlinks,%
            citecolor=blue,%
            filecolor=black,%
            linkcolor=black,%
            urlcolor=blue}

\newtheorem{theoreme}{Theorem}[section]

\newtheorem{corollaire}{Corollary}
\newtheorem{lemme}{Lemma}

\newtheorem*{acknowledgements}{Acknowledgements}

\newcommand{\C}{\mathbb{C}ov}

\newcommand{\E}{\mathbb{E}}
\newcommand{\V}{\mathbb{V}}

\newcommand{\R}{\mathbb{R}}
\newcommand{\p}{\partial}

\newcommand{\X}{\mathcal{X}}

\newcommand{\hpsi}{\theta}

\begin{document}
\date{}
\title{ Propagation of initial errors on the parameters for linear and Gaussian state space models}
\author{Salima El Kolei}
\maketitle
\footnotetext[1]{
Affiliation: S. El Kolei,
              Laboratoire de Math\'{e}matiques J.A. Dieudonn\'{e}\\
              UMR n 7351 CNRS UNSA\\
              Universit\'{e} de Nice - Sophia Antipolis\\
              06108 Nice Cedex 02 France \\
              Tel.: +33-04-92-07-62-56\\
 \href{salima@unice.fr}{salima@unice.fr}
}

\begin{abstract}
For linear and Gaussian state space models parametrized by $\theta_0  \in \Theta \subset \R^{r}, r \geq 1$ corresponding to the vector of parameters of the model, the Kalman filter gives exactly the solution for the optimal filtering under weak assumptions. This result supposes that $\theta_0$ is perfectly known. In most real applications, this assumption is not realistic since $\theta_0$ is unknown and has to be estimated. In this paper, we analysis the Kalman filter for a biased estimator $\theta$ of $\theta_0$. We show the propagation of this bias on the estimation of the hidden state. We give an expression of this propagation for linear and Gaussian state space models and we extend this result for almost linear models estimated by the Extended Kalman filter. An illustration is given for the autoregressive process with measurement noises widely studied in econometrics to model economic and financial data.\\

\end{abstract}

\section{Introduction}
\label{intro}

Let $(\Omega,\mathcal{F},\mathbb{P}_{\theta_0})$ be a probability space parametrized by $\theta_0  \in \Theta \subset \R^{r}, r \geq 1$ corresponding to the vector of parameters of the model. We define two real vectors $\left\{x_t, t\in \mathbb{N}\right\}$ defined on $(\Omega,\mathcal{F},\mathbb{P}_{\theta_0})$ with value in $\X$ and $\left\{y_t, t\in \mathbb{N}^{*}\right\}$ defined on $(\Omega,\mathcal{F},\mathbb{P}_{\theta_0})$ with value in $\mathcal{Y}$. The process $\left\{x_t, t\in \mathbb{N}\right\}$ (respectively $\left\{y_t, t\in \mathbb{N}^{*}\right\}$) is called the unobserved signal process (resp. the observation process).\\ 

The Kalman filter (KF) and the Extended Kalman filter (EKF) commonly used in some engineering applications have been successfully employed in various areas. These filters may be easily understood by reading the first publication of Kalman in 1960 \cite{Ka} or the Bayesian interpretation of Harrison and Stevens in 1971 \cite{harisson}.

\subsection{The (Extended) Kalman filter: motivation}

Let $y_1, \cdots, y_T$ be the data (which may be either a scalar or a vector) at time $1,\cdots,T$. We assume that $y_t$ depends on the unobservable variable $x_t$. The aim of the (Extended) Kalman filter is to make inference about the hidden state $x_t$ (which may also be a scalar or a vector) conditionally to the data $y_1,\cdots, y_t$. The relationship between the observed variable $y_t$ and the hidden state $x_t$ is linear and described by a function $h$ depending on the unknown vector of parameters $\theta_0$. This relation is specified by the following observation equation:

\begin{equation*}
y_{t}=h(\theta_0, x_{t})+ \sigma^{\varepsilon}_{\theta_0}\varepsilon_t, \qquad t\geq 1
\end{equation*} 
where $\varepsilon_t$ is the vector of noises assumed to be normally distributed with mean zero and unit variance, denoted as: $\varepsilon_t \sim \mathcal{N}(0, I_{n_y\times n_y})$ where $n_y$ is the dimension of the observation space $\mathcal{Y}$. \\
The hidden state $x_t$ is assumed to be varying with time and its dynamic feature is given by the following state equation:

\begin{equation*}
x_t=b(\theta_0, x_{t-1})+\sigma^{\eta}_{\theta_0}\eta_t
\end{equation*} 
where $b$ is a known function and $\eta_t$ is the state error assumed to be normally distributed with mean zero and unit variance, i.e $\eta_t \sim \mathcal{N}(0, I_{n_x\times n_x})$ where $n_x$ is the dimension of the state space $\mathcal{X}$.\\
In addition to the usual Kalman filter assumptions (see \cite{Ka}), we also assume that the noises $\varepsilon_t$ and $\eta_t$ are independent. \\
Hence, this paper is concerned with the following discrete time state space model with additives noises:

\begin{equation}\label{kf_erreur}
\left\lbrace\begin{array}{ll}
y_{t}=h(\theta_0, x_{t})+ \sigma^{\varepsilon}_{\theta_0}\varepsilon_t, \qquad t\geq 1\\
x_t=b(\theta_0, x_{t-1})+\sigma^{\eta}_{\theta_0}\eta_t
\end{array}
\right.
\end{equation} 

Under the usual Kalman assumptions, the model (\ref{kf_erreur}) can be rewritten as follows:

\begin{equation}\label{linearisemodel}
\left\lbrace\begin{array}{ll}
y_{t}=d_{t}(\theta_0)+C_{\theta_0}x_{t}+\sigma^{\varepsilon}_{\theta_0}\varepsilon_{t}\quad t\geq 1,\\
x_{t}=u_{t}(\theta_0)+A_{\theta_0}x_{t-1}+\sigma^{\eta}_{\theta_0}\eta_{t}, 
\end{array}
\right.
\end{equation}

$\newline$
If the vector of parameters $\theta_0$ is perfectly known, the optimal filtering $p_{\theta_0}(x_t|y_{1:t})$ is Gaussian and the Kalman filter gives exactly the two first conditional moments: $\hat{x}_{t}=\E_{\theta_0}[x_t \vert y_{1:t}]$ and $P_t=\E_{\theta_0}[(x_t-\hat{x}_{t})(x_t-\hat{x}_{t})' \vert y_{1:t}]$ where $^{'}$ stands for the transpose. In particular, the Kalman filter estimator is the \emph{BLUE} (Best Linear and Unbiased Estimator) among linear estimators. Nevertheless, in most applications the linearity assumption of the functions $h$ and $b$ is not always satisifed. A linearization by a one order Taylor series expansion can be performed and the Extended Kalman filter consists in applying the Kalman filter on this linearized model.\\
For the EKF, the matrix $C_{\theta_0}$ is the differential of the function $h$ with respect to (w.r.t.) $x$ computed at the point $(\theta_0, \hat{x}_t^{-})$ where $\hat{x}_t^{-}$ corresponds to the conditional expectation $\E_{\theta_0}[x_t \vert y_{1:t-1}]$. Additionally, the matrix $A_{\theta_0}$ is the differential of the function $b$ w.r.t. $x$ computed at the point $(\theta_0, \hat{x}_{t-1})$. Furthermore, the functions $u_t(\theta_0)$ and  $d_t(\theta_0)$ are defined as:

\begin{equation*}
\left\lbrace\begin{array}{ll}
u_{t}(\theta_0)=b(\theta_0, \hat{x}_{t-1})-A_{\theta_0}\hat{x}_{t-1}\\
d_{t}(\theta_0)=h(\theta_0, \hat{x}^{-}_{t})-C_{\theta_0}\hat{x}^{-}_{t}
\end{array}\right.
\end{equation*}

In this paper, we assume that the vector of parameters $\theta_0$ is not perfectly known such that the inference of the hidden state $x_t$ conditionnally to $y_{1:t}$ is made with errors of specification. This typical case is frequent in practice since in general the vector of parameters is unknown and need to be estimated by an ordinary method. The resulting estimator can be biased and consequently this bias is propagated on the estimation of the hidden state. More precisely, if we denote by $\hat{\theta}$ a biased estimator of $\theta_0$ such that $\E_{\theta_0}[\hat{\theta}]=\theta=\theta_0+\epsilon$ where $\epsilon$ is a fixed and unknown error corresponding to the bias, we want to evaluate the propagation of the error a posteriori and of the residues a posteriori given by:
\begin{eqnarray}\label{def}
e_{t}=x_{t}-\E_{\theta}[x_t \vert y_{1:t}] \\
\xi_{t}=y_{t}-\E_{\theta}[y_t \vert y_{1:t}].
\end{eqnarray}

Many papers concerned the propagation of the initial error on the state $(x_{0}-\hat{x}_{0})$ through the filter, and, to the best of our knowledge, there don't exist in the literature, an analysis of the propagation of the initial errors on the vector of parameters. 
In this paper, we derive an expression of these propagations for the Kalman and the Extended Kalman filters. Our main result shows that a correlation between the error a posteriori $e_{t}$ and the unobserved state $x_t$ appeared at each time $t$ of the filter. The Kalman filter is now a biased estimator and a new Lyapunov dynamic equation for the variance matrix $P_t$ is induced. \\

Applications of this result include epidemiology, meteorology, neuroscience, ecology (see \cite{MR2850220}) and finance (see \cite{johannes}). For example, our result can be applied to the five ecological state space models described in \cite{peters}.  Although the scope of our method is general, we have chosen to focus on the so-called autoregressive process AR(1) with measurement noise which has been widely studied and on which our main result can be easily applied and understood.
A full illustration of this result is given for a more complex model as the Heston model which is very used in finance for pricing options and hedging portfolios (see \cite{salima2}).\\

$\newline$
The paper is organized as follows. Section \ref{main result} presents the model assumptions and states all of the theoretical results. The application is given and discussed in Section \ref{exemples}. Some concluding remarks are provided in the last section. The proofs are gathered in Appendix \ref{appendice}. 
 
 \newpage

\section{Main result}\label{main result}

\subsection{General setting and assumptions}\label{kalmannotat}

In this section, we introduce some preliminary main notations and provide the assumptions of model (\ref{linearisemodel}).

\subsubsection{Notations}

Subsequently, we denote by $E_{t}$ the pair of the unobservable states vectors given by
                 $\left(\begin{array}{cc}
                 e_{t}\\
                 x_{t}\\
                 \end{array}
                 \right)$ and by $\mathfrak{E}_t$ the pair of the observations vectors $\left(\begin{array}{cc}
                 \xi_{t}\\
                 y_{t}\\
                 \end{array}
                 \right)$ where $e_{t}$ and $\xi_{t}$ are defined in (\ref{def}) respectively. Their variances matrix are denoted by $\Sigma^x_t$ and $\Sigma^{y}_t$ respectively.\\
                 
Regarding the partial derivatives, for any function $h$, $[\p h/\p \theta]$ is the vector of the partial derivatives of $h$ w.r.t $\theta$.\\ 
 
Finally, $R_{\theta_0}$ denotes $\sigma_{\theta_0}^{\varepsilon}\sigma_{\theta_0}^{'\varepsilon}$ and $Q_{\theta_0}$ denotes $\sigma^{\eta}_{\theta_0}\sigma^{'\eta}_{\theta_0}$ and are the covariances matrix of $\varepsilon_t$ and of $\eta_t$ respectively.
%\begin{remarque}
%Moreover, the following derivatives $[\p b/\p \theta]$ (respectively $[\p h/\p \theta]$) are computed at the points $(\theta, \hat{x}_{t-1})$ (respectively $(\theta, \hat{x}_t^{-})$) and the derivatives $[\p \sigma_{\theta}^{\varepsilon}/\p \theta]$ and $[\p \sigma_{\theta}^{\eta}/\p \theta]$ are computed at the point $\theta$.\\
%\end{remarque}

\subsubsection{Assumptions}

We consider the state space models (\ref{linearisemodel}), the following assumption ensures some smoothness for the functions $h$ and $b$.\\

 \textbf{(A1)} The functions $b$ and $h$ are differentiable with respect to $\theta_0$ and $x$. \\

\subsubsection{Main result}
Before running into the main theorem of this paper, let us explain some existing results. It is well known that if the vector of parameters is exactly known, the error a posteriori $e_t$ is given by the following formula:

\begin{equation}
e_t =(I_{n_x\times n_x}-K_t C_{\theta_0})A_{\theta_0}e_{t-1}-K_t( \sigma^{\varepsilon}_{\theta_0}\varepsilon_t+ C_{\theta_0} \sigma^{\eta_0}_{\theta_0}\eta_t)+ \sigma^{\eta}_{\theta_0}\eta_t\nonumber\\
\end{equation}
where $K_t$ is called the Kalman matrix that minimizes the variance matrix $P_t$.\\
Under some assumptions on the model (\ref{linearisemodel}), a CLT is obtained for $e_t$ as $t$ tends to infinity (see \cite{AbCm94}).
The following Theorem gives the propagation of the error a posteriori $e_t$ and of $\xi_t$ for the Kalman filter and the Extended Kalman filter when $\theta_0$ is not exactly known. In this respect, we further assume that assumption \textbf{(A1)} holds true and the usual Kalman assumptions are satisfied.

\begin{theoreme} \label{prop_ekf} Consider the model (\ref{linearisemodel}). If $\epsilon <<1$, then:

\begin{eqnarray} 
e_t &=&(I_{n_x\times n_x}-K_t C_{\theta})A_{\theta}e_{t-1}-K_t( \sigma^{\varepsilon}_{\theta}\varepsilon_t+ C_{\theta} \sigma^{\eta}_{\theta}\eta_t)+ \sigma^{\eta}_{\theta}\eta_t\nonumber\\
&&+
\mathcal{E}_{x}^{\epsilon}(\theta,t)+\mathcal{F}_{x}^{\epsilon}(\theta,t)x_{t-1}+\mathcal{W}_{x}^{\epsilon}(\theta,t)+o(\epsilon)\label{f1}
\end{eqnarray} 
with:

\begin{eqnarray} 
&&\mathcal{E}_{x}^{\epsilon}(\theta,t)=-\epsilon \left( (I_{n_x\times n_x}-K_t C_{\theta})\frac{\p u_t}{\p \theta}(\theta)-K_t\frac{\p d_t}{\p \theta}(\theta)-K_t\frac{\p C_{\theta}}{\p \theta}u_t(\theta)\right)\label{matrixA1}\\
&&\mathcal{F}_{x}^{\epsilon}(\theta,t)=-\epsilon \left( (I_{n_x\times n_x}-K_t C_{\theta})\frac{\p A_{\theta}}{\p \theta}-K_t\frac{\p C_{\theta}}{\p \theta}A_{\theta}\right)\label{matrixA2}\\
&&\mathcal{W}_{x}^{\epsilon}(\theta,t)=-\epsilon \left( \frac{\p  \sigma^{\eta}_{\theta}}{\p \theta}\eta_t-K_t C_{\theta} \frac{\p  \sigma^{\eta}_{\theta}}{\p \theta}\eta_t-K_t  \sigma^{\eta}_{\theta} \frac{\p C_{\theta}}{\p \theta}\eta_t-K_t \frac{\p  \sigma^{\varepsilon}_{\theta}}{\p \theta}\varepsilon_t\right)\label{matrixA3}
\end{eqnarray}
Additionally, the propagation of $\xi_t$ is equal to:

\begin{equation}\label{f2}
\xi_t=C_{\theta} e_{t}+  \sigma^{\varepsilon}_{\theta}\varepsilon_t+
\mathcal{E}_{y}^{\epsilon}(\theta,t)+\mathcal{F}_{y}^{\epsilon}(\theta,t)x_{t}+\mathcal{W}_{y}^{\epsilon}(\theta,t)+o(\epsilon)
\end{equation}
with:

\begin{eqnarray} \label{matrixC}
\mathcal{E}_{y}^{\epsilon}(\theta,t)=-\epsilon\frac{\p d_t}{\p \theta}(\theta),\quad \mathcal{F}_{y}^{\epsilon}(\theta,t)=-\epsilon \frac{\p C_{\theta}}{\p \theta},\quad \mathcal{W}_{y}^{\epsilon}(\theta,t)=-\epsilon\frac{\p  \sigma^{\varepsilon}_{\theta}}{\p \theta}\varepsilon_t
\end{eqnarray}

Moreover, when the linearity assumption of the funtions $b$ and $h$ is not satisfied, the formulas above remain true with the notations of the EKF defined in Section \ref{intro}.
\end{theoreme}
\begin{proof}
See Appendix (\ref{preuves1}).
\end{proof}

$\newline$
We note that the terms depending on $\epsilon$: $\mathcal{E}_{x}^{\epsilon}(\theta,t)$, $\mathcal{F}_{x}^{\epsilon}(\theta,t)$  and $\mathcal{W}_{x}^{\epsilon}(\theta,t)$ (resp. $\mathcal{E}_{y}^{\epsilon}(\theta, t)$, $\mathcal{F}_{y}^{\epsilon}(\theta,t)$ and $\mathcal{W}_{y}^{\epsilon}(\theta,t)$) are the corrective terms arising from the bias of the parameters estimates.\\
Besides, we can see in Eq.(\ref{f1}) that at time $t$, the propagation of the state error $e_{t}$ depends on $e_{t-1}$ but also on the true state variable $x_{t-1}$. Therefore, the variance of the error $e_t$ depends on the variance of $e_{t-1}$ but also on the covariance between $e_{t-1}$ and $x_{t-1}$.\\

Theorem \ref{prop_ekf} gives an expression of the error a posteriori $e_t$ and of the residues $\xi_t$ which can be rewritten as follows:

\begin{eqnarray*}
e_t&=& x_t-\E_{\theta}[x_t \vert y_{1:t}]\\
&=&\underbrace{(x_t-\E_{\theta_0}[x_t \vert y_{1:t}])}_{ \text{ (1) error of estimation}}+\underbrace{(\E_{\theta_0}[x_t \vert y_{1:t}]-\E_{\theta}[x_t \vert y_{1:t}])}_{ \text{ (2) correctives terms arising from the bias of parameters.}}\\
\end{eqnarray*} 
Additionally,
\begin{eqnarray*}
\xi_t&=& x_t-\E_{\theta}[\xi_t \vert y_{1:t}]\\
&=&\underbrace{(\xi_t-\E_{\theta_0}[\xi_t \vert y_{1:t}])}_{  \text{ (1) true residues}}+\underbrace{(\E_{\theta_0}[\xi_t \vert y_{1:t}]-\E_{\theta}[\xi_t \vert y_{1:t}])}_{\text{ (2) correctives terms arising from the bias of parameters.}}\\
\end{eqnarray*}

The expression of the correctives terms (2) are given in Corollary \ref{coroll_estierorr1}.\\

\begin{corollaire}\label{coroll_estierorr1} Let $e_t$ given in Eq.(\ref{f1}), the mean of the error a posteriori is given by:
\begin{eqnarray}
\E_{\theta_0}[e_t \vert y_{1:t}]&=&\E_{\theta_0}[x_t \vert y_{1:t}]-\E_{\theta}[x_t \vert y_{1:t}]\nonumber\\
&=&\big((I_{n_x\times n_x}-K_tC_{\theta})A_{\theta} + \mathcal{F}_{x}^{\epsilon}(\theta,t)\big)\E_{\theta_0}[e_{t-1} \vert y_{1:t-1}]\nonumber\\
&&+\mathcal{E}^{\epsilon}_x(\theta,t)+\mathcal{F}_{x}^{\epsilon}(\theta,t)\E_{\theta}[x_{t-1}\vert y_{1:t-1}]+o(\epsilon)\label{estimation_error}
\end{eqnarray}
where $\mathcal{E}^{\epsilon}_x(\theta,t)$ and $\mathcal{F}^{\epsilon}_x(\theta,t)$ are given in Eq.(\ref{matrixA1}) and Eq.(\ref{matrixA2}).\\

Besides, let $\xi_t$ given in Eq.(\ref{f2}), the mean of $\xi_t$ is given by:

\begin{eqnarray}
\E_{\theta_0}[\xi_t \vert y_{1:t}]&=&\E_{\theta_0}[\xi_t \vert y_{1:t}]-\E_{\theta}[\xi_t \vert y_{1:t}]\nonumber\\
&=&\left(C_{\theta}+\mathcal{F}_{y}^{\epsilon}(\theta,t)\right)\E_{\theta_0}[e_{t} \vert y_{1:t}]+\mathcal{E}^{\epsilon}_y(\theta,t)+\mathcal{F}_{y}^{\epsilon}(\theta,t)\E_{\theta}[x_t\vert y_{1:t}]+o(\epsilon)\label{estimation_error2}
\end{eqnarray}
where $\mathcal{E}^{\epsilon}_y(\theta,t)$ and $\mathcal{F}^{\epsilon}_y(\theta,t)$ are given in Eq.(\ref{matrixC}).
\end{corollaire}

$\newline$

Corollary \ref{coroll_estierorr1} which is just a consequence of Theorem (\ref{prop_ekf}) gives a computable recursive expression of the expected error $\E_{\theta_0}[e_t \vert y_{1:t}]$. Given $\E_{\theta_0}[e_0]$ one can deduce all the values of this expectation for all $t=1,\cdots, T$.\\

\section{Illustration on the linear Gaussian AR(1) model:}\label{exemples}

\subsection{The model}
Let us consider the linear AR(1) model with measurement noise given by:
\begin{equation}\label{AR}
\left\lbrace\begin{array}{ll}
y_t=x_{t}+\sigma_{\varepsilon}\varepsilon_{t}, \qquad t=1,\cdots, T\\
x_{t+1}=\phi_0 x_{t}+\sigma_{\eta}\eta_{t+1}.
\end{array}
\right.
\end{equation}
Since this model is linear and Gaussian we can apply Eq.(\ref{estimation_error}) in Corollary \ref{coroll_estierorr1} to recover the expectation of $e_t$ when the state $x_t$ is estimated with a biased vector of parameters. For this straighforward example, $\theta_0$ is equal to $\phi_0$. For the simulation, we take $\phi_0=0.7$, $\sigma^{2}_{\eta}=0.3$ and $\sigma^{2}_{\varepsilon}=0.5$.

\subsection{Numerical result:} We run a Kalman filter by assuming that the parameter estimate $\phi$ is biased and we take $\phi=0.85$, that is $\epsilon=0.15$. For this model, the functions $b$ and $h$ are given by:
 
\begin{equation*}
b(\theta_0, x)= \phi_{0}x \text{ and } h(\theta_0, x)= x
\end{equation*}
The variable $A_{\theta_0}$ is equal to $\phi_0$ and $C_{\theta_0}$ is equal to one. The control variables $u_t(\theta_0)$ and $d_t(\theta_0)$ are equal to zero.\\
 Furthermore, the functions $\mathcal{E}^{\epsilon}_{x}(\theta, t), \mathcal{F}^{\epsilon}_{x}(\theta, t)$ are easily computable and given in the following lemma.

\begin{lemme} For the linear AR(1) model, the functions $\mathcal{E}^{\epsilon}_{x}(\theta, t)$ and $\mathcal{F}^{\epsilon}_{x}(\theta, t)$ are equal to: \begin{equation*}
\mathcal{E}^{\epsilon}_{x}(\theta, t)=0, \text{ and }  \mathcal{F}^{\epsilon}_{x}(\theta, t)=-\epsilon (1-K_t)
\end{equation*}
Therefore, by using Eq.(\ref{estimation_error}) of Corollary \ref{coroll_estierorr1}, the expectation $\E_{\theta_0}[e_t \vert y_{1:t}]$ is given by

\begin{equation}\label{formAR}
\E_{\theta_0}[e_t \vert y_{1:t}]=\big((1-K_t)(\phi -\epsilon) \big)\E_{\theta_0}[e_{t-1} \vert y_{1:t-1}]-\epsilon (1-K_t)\E_{\theta}[x_{t-1}\vert y_{1:t-1}]+o(\epsilon)
\end{equation}
\end{lemme}

\begin{figure}[H]
\includegraphics[width=129mm, height=60mm]{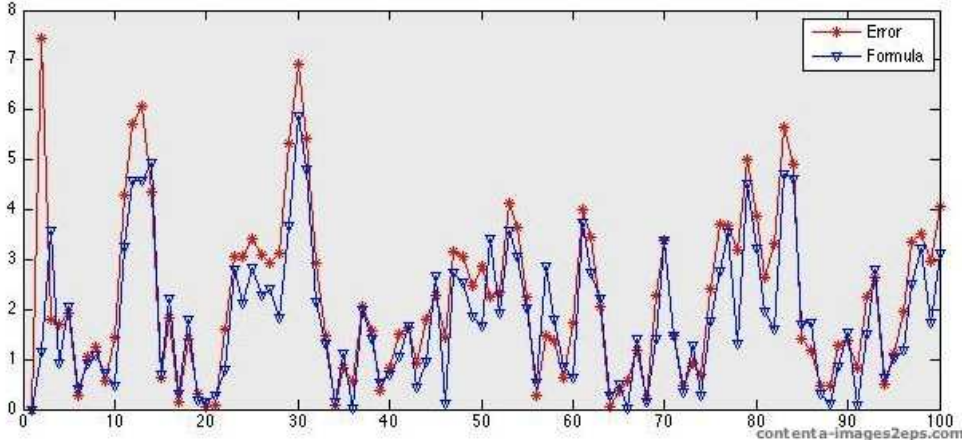}
\caption{\normalsize  Red: True Error $\E_{\theta_0}[e_t\vert y_{1:t}]\times 100$. Blue: Approximation (\ref{formAR})$\times 100$. }
\label{compar_par_phi_kf} % Give a unique label
\end{figure}

$\newline$

This example shows that the approximation (\ref{formAR}) explains the true error $\E_{\theta_0}[e_t \vert y_{1:t}]$ for an easy model. The term $\mathcal{F}_{x}^{\epsilon}(\theta, t)$ corresponds to the bias of $x_t$ induced by the bias of the parameter estimate. Furthermore, we can see that the error between the true expectation $\E_{\theta_0}[e_t \vert y_{1:t}]$ and the approximation (\ref{formAR}) corresponds to $o(\epsilon)$. A full application is given in \cite{salima2}.\\

The following Theorem regards the expression of the variances matrix $\Sigma^x_t$ and $\Sigma^{y}_t$ of $E_t$ and $\mathfrak{E}_t$ respectively.\\
    
\begin{theoreme}\label{prop}

The variance matrix $\Sigma_t^x$ is given by: 

\begin{equation*}
\Sigma_t^x= \begin{pmatrix}
			           V^{x}_{t} & S^{'x}_{t}\\
		 S^{x}_{t} & P^{x}_{t}
			          \end{pmatrix}
\end{equation*}

where:
\begin{eqnarray*}
&&V^{x}_{t} = (I-K_tC_{\hpsi})A_{\hpsi} V^{x}_{t-1}A^{'}_{\hpsi} (I-K_tC_{\hpsi})^{'}+\mathcal{F}_{x}^{\epsilon}(\theta) S_{t-1} (I-K_tC_{\hpsi})A^{'}_{\hpsi} \\
&& \qquad +(I-K_tC_{\hpsi})A_{\hpsi} S^{'x}_{t-1}\mathcal{F}_{x}^{'\epsilon}(\theta)+\mathcal{F}_{x}^{\epsilon}(\theta)  P^{x}_{t-1} \mathcal{F}_{x}^{'\epsilon}(\theta) + \V_{\theta}[\mathcal{\tilde{W}}_{x}^{\epsilon}(\theta)]\\
&&S^{x}_t = A_{\theta_0} S^{x}_{t-1} A^{'}_{\theta_0}(I-K_tC_{\hpsi})^{'}+ A_{\theta_0} P^{x}_{t-1}\mathcal{F}_{x}^{'\epsilon}(\theta)+Cov_{\theta}\left( \mathcal{\tilde{W}}_{x}^{\epsilon}(\theta),\sigma^{\eta}_{\theta_0} \eta_{t} \right)\\
&&P^{x}_t = A_{\theta_0} P^{x}_{t-1} A_{\theta_0}^{'} + Q_{\theta_0} 
\end{eqnarray*}

with:

\begin{equation}\label{vara}
\mathcal{\tilde{W}}_{x}^{\epsilon}(\theta)=\mathcal{W}_{x}^{\epsilon}(\theta)+ \sigma^{\eta}_{\theta}\eta_t-K_t\sigma^{\varepsilon}_{\theta}\varepsilon_t -K_t C_{\theta} \sigma^{\eta}_{\theta}\eta_t
\end{equation}

where $\mathcal{E}_{x}^{\epsilon}(\theta)$, $\mathcal{F}_{x}^{\epsilon}(\theta)$ and $\mathcal{W}_{x}^{\epsilon}(\theta)$ are given in Eq.(\ref{matrixA1}), Eq.(\ref{matrixA2}) and Eq.(\ref{matrixA3}) in Theorem \ref{prop_ekf}.\\

If $\epsilon<<1$, then $\V_{\theta}[\mathcal{\tilde{W}}_{x}^{\epsilon}(\theta)]=Q_{\theta} +K_t\left(C_{\theta}Q_{\theta}C_{\theta}'+R_{\theta}\right)K_{t}'$ and $Cov_{\theta}\left( \mathcal{\tilde{W}}_{x}^{\epsilon}(\theta),\sigma^{\eta}_{\theta_0} \eta_{t} \right)$ is given by:

\begin{equation*}
Cov_{\theta}\left( \mathcal{\tilde{W}}_{x}^{\epsilon}(\theta),\sigma^{\eta}_{\theta_0} \eta_{t} \right)=-\epsilon\left(\frac{\p \sigma^{\eta}_{\theta} }{\p \theta} Q_{\theta_0}\sigma^{'\eta}_{\theta_0}-K_t \left(C_{\theta}\frac{\p \sigma^{\eta}_{\theta} }{\p \theta}+\sigma^{\eta}_{\theta}\right)Q_{\theta_0}\sigma^{'\eta}_{\theta_0}\right)
\end{equation*}
		
Additionally, the variance matrix $\Sigma_t ^{y}$ is given by

\begin{equation*}
\Sigma_t ^{y}=\begin{pmatrix}
			 V_t ^{y} & S ^{'y}_t\\
			 S_t ^{y} & P_t ^{y}
			 \end{pmatrix}
\end{equation*}					

where:
\begin{eqnarray*}
&&V_t ^{y} = C_{\hpsi} V^{x}_{t} C_{\hpsi}^{'} +\mathcal{F}_{y}^{\epsilon}(\theta) S_{t} C_{\hpsi}^{'}+ C_{\hpsi}S^{'}_{t}\mathcal{F}_{y}^{'\epsilon}(\theta)+\mathcal{\tilde{F}}_{y}^{\epsilon}(\theta) P^{x}_{t}\mathcal{\tilde{F}}_{y}^{\epsilon}(\theta)+ \V_{\theta}[\mathcal{\tilde{W}}_{y}^{\epsilon}(\theta)]\\
&&S_t ^{y} = C_{\theta_0}S^{x}_{t} C_{\hpsi}^{'} +C_{\theta_0} P^{x}_{t}\mathcal{F}_{y}^{'\epsilon}(\theta)+Cov_{\theta}\left( \mathcal{\tilde{W}}_{y}^{\epsilon}(\theta),\sigma^{\varepsilon}_{\theta_0}\varepsilon_{t} \right)\\
&&P_t ^{y} = C_{\theta_0} P^{x}_{t} C_{\theta_0}^{'} + R_{\theta_0} 
\end{eqnarray*}

with:

\begin{equation}\label{vara2}
\mathcal{\tilde{W}}_{y}^{\epsilon}(\theta)=\mathcal{W}_{y}^{\epsilon}(\theta)+\sigma^{\varepsilon}_{\theta}\varepsilon_t
\end{equation}
where $\mathcal{E}_{y}^{\epsilon}(\theta)$, $\mathcal{F}_{y}^{\epsilon}(\theta)$ and $\mathcal{W}_{y}^{\epsilon}(\theta)$ are given in Eq.(\ref{matrixC}) Theorem \ref{prop_ekf}.\\

If $\epsilon << 1$, then $\V_{\theta}[\mathcal{\tilde{W}}_{y}^{\epsilon}(\theta)]=R_{\theta}$ and $Cov_{\theta}\left( \mathcal{\tilde{W}}_{y}^{\epsilon}(\theta),\sigma^{\varepsilon}_{\theta_0}\varepsilon_{t} \right)$ is given by:

\begin{equation*}
Cov_{\theta}\left( \mathcal{\tilde{W}}_{y}^{\epsilon}(\theta),\sigma^{\varepsilon}_{\theta_0}\varepsilon_{t} \right)=-\epsilon \left(\frac{\p \sigma^{\varepsilon}_{\theta} }{\p \theta}R_{\theta_0}\sigma^{'\varepsilon}_{\theta_0}\right)
\end{equation*}
\end{theoreme}

\begin{proof}
See Appendix (\ref{preuves2}).
\end{proof}

$\newline$
The quantities $\mathcal{F}_{x}^{\epsilon}(\theta,t)$  and $\mathcal{W}_{x}^{\epsilon}(\theta,t)$ (resp. $\mathcal{F}_{y}^{\epsilon}(\theta,t)$ and $\mathcal{W}_{y}^{\epsilon}(\theta,t)$) correspond to the correctives terms arising from the bias of the parameters and in particular from the correlation between $e_t$ and the true state $x_t$ (see Eq.(\ref{f1})). This correlation induces a new Lyapunov dynamic equation for the variance matrix $V^{x}_t$. For unbiased parameters estimates, these terms are dropped and a CLT is given in \cite{AbCm94}.\\

\section{Concluding remarks and discussion}
In this paper we provide an expression of the propagation errors on the hidden state for an initial and fixed error on the vector of parameters.\\
We showed that the hidden state $x_t$ appaered in the propagation equation inducing a correlation between $e_t$ and the true state $x_t$ and most importantly a new Lyapunov dynamic equation for the variance matrix. By using the same assumptions than in \cite{AbCm94} and adding smoothness assumptions on the functions $b$ and $h$ and on their derivatives, one can again obtain a CLT for $e_t$. Nevertheless, it is not the subject of this paper.\\
Another remark concerns the case where $\epsilon$ is not fixed and is supposed to be a random variable. This particular case refers to the approach proposed in \cite{MR1847791} for which the parameters are supposed time varying. A dynamical artificial evolution is assumed for $\theta$ such that $\theta_t=\theta_{t-1}+\sigma^{Z}Z$ where $Z$ is a centered and standard gaussian random variable. To the best of our knowledge, there does not exist results about the convergence of this approach. This method fails in practice when the variance $\sigma^{Z}$ is not small. Some authors use $\sigma^{Z}$ decreasing with time. Hence, at each step of the filter, a small perturbation is added to the parameters. This can be seen as a small bias $\epsilon$ introduced at the first step of the filter.

\newpage

\appendix \label{appendice}

\small
\section{\label{preuves1} Proof of Theorem \ref{prop_ekf}:}\quad\\

The proof is essentially based on a one order Taylor expansion of the functions $b$ and $h$ with respect to $\theta$.
\begin{eqnarray}
e_t=x_t-\E_{\theta}[x_t|y_{1:t}]&=&x_t-\E_{\theta}[x_{t}|y_{1:t-1}]-K_t(y_t-\hat{y}_t^{-})\nonumber\\
&=&u_{t}(\theta_0)+A_{\theta_0}x_{t-1}+ \sigma^{\eta}_{\theta_0}\eta_t-\E_{\theta}[u_{t}(\theta)+A_{\theta}x_{t-1}+ \sigma^{\eta}_{\theta}\eta_t\vert y_{1:t-1}]-K_t(y_t-\hat{y}_t^{-})\nonumber\\
&=&u_{t}(\theta_0)+A_{\theta_0}x_{t-1}+ \sigma^{\eta}_{\theta_0}\eta_t-\E_{\theta}[u_{t}(\theta)+A_{\theta}x_{t-1}\vert y_{1:t-1}]-K_t(y_t-\hat{y}_t^{-})\nonumber\\
&=&u_{t}(\theta_0)+A_{\theta_0}x_{t-1}+ \sigma^{\eta}_{\theta_0}\eta_t-u_{t}(\theta)-A_{\theta}\E_{\theta}[x_{t-1}|y_{1:t-1}]-K_t(y_t-\hat{y}_t^{-})\nonumber\\
&=&u_{t}(\theta_0)+A_{\theta_0}x_{t-1}+ \sigma^{\eta}_{\theta_0}\eta_t-u_{t}(\theta)-A_{\theta}\hat{x}_{t-1}-K_t(y_t-\hat{y}_t^{-})\label{erreur_ekf1_inter}
\end{eqnarray}
Note that one can write:

\begin{eqnarray*}
&u_t(\theta_0) = u_t(\theta)-\epsilon \frac{\p u}{\p \theta}(\theta)+o(\epsilon),\quad A_{\theta_0} = A_{\theta}-\epsilon \frac{\p A_{\theta}}{\p \theta}+o(\epsilon),\quad \sigma^{\eta}_{\theta_0}= \sigma^{\eta}_{\theta_0}-\epsilon \frac{\p  \sigma^{\eta}_{\theta}}{\p \theta}+o(\epsilon)
\end{eqnarray*}
Pluging into (\ref{erreur_ekf1_inter}), one gets:
\begin{eqnarray}
e_t=A_{\theta}e_{t-1}-\epsilon\frac{\p u_t}{\p \theta}(\theta)-\epsilon\frac{\p A_{\theta}}{\p \theta}x_{t-1}+ \sigma^{\eta}_{\theta}\eta_t-\epsilon \frac{\p  \sigma^{\eta}_{\theta}}{\p \theta}\eta_t-K_t(y_t-\hat{y}_t^{-})+o(\epsilon)\label{erreur_ekf1}
\end{eqnarray}
Furthermore,

\begin{eqnarray}
(y_t-\hat{y}_t^{-})&=&d_{t}(\theta_0)+C_{\theta_0}x_{t}+ \sigma^{\varepsilon}_{\theta_0}\varepsilon_t-\E_{\theta}[y_{t}|y_{1:t-1}]\nonumber\\
&=&d_{t}(\theta_0)+C_{\theta_0}x_{t}+ \sigma^{\varepsilon}_{\theta_0}\varepsilon_t-d_{t}(\theta)-C_{\theta}\E_{\theta}[x_{t}|y_{1:t-1}]\nonumber
\end{eqnarray}
and 
\begin{eqnarray*}
d_t(\theta_0) = d_t(\theta)-\epsilon \frac{\p d}{\p \theta}(\theta)+o(\epsilon),\quad C(\theta_0) = C(\theta)-\epsilon \frac{\p C}{\p \theta}(\theta)+o(\epsilon),\quad \sigma^{\varepsilon}_{\theta_0}= \sigma^{\varepsilon}_{\theta_0}-\epsilon \frac{\p  \sigma^{\varepsilon}_{\theta}}{\p \theta}+o(\epsilon)
\end{eqnarray*}
So that:

\begin{eqnarray}
(y_t-\hat{y}_t^{-})&=&-\epsilon \frac{\p d_{t}}{ \p \theta}(\theta)+( \sigma^{\varepsilon}_{\theta}-\epsilon \frac{\p  \sigma^{\varepsilon}_{\theta}}{\p \theta})\varepsilon_t + (C_{\theta}-\epsilon \frac{\p C_{\theta}}{\p \theta})\left(u_{t}(\theta_0)+A_{\theta_0}x_{t-1}+ \sigma^{\eta}_{\theta_0}\eta_t \right)\nonumber\\
&&-C_{\theta}\E_{\theta}[x_{t}|y_{1:t-1}]+o(\epsilon)\nonumber
\end{eqnarray}
Rewrite,
\begin{equation*}
\E_{\theta}[x_{t}|y_{1:t-1}] = A_{\theta}\E_{\theta}[x_{t-1}|y_{1:t-1}]+u_t(\theta)
\end{equation*}
we get:
\begin{eqnarray}
&=&C_{\theta}A_{\theta}x_{t-1}-C_{\theta}A_{\theta}\hat{x}_{t-1}+ \sigma^{\varepsilon}_{\theta}\varepsilon_t+C_{\theta} \sigma^{\eta}_{\theta}\eta_t\nonumber\\
&&-\epsilon\left(\frac{\p d_{t}}{ \p \theta}(\theta)+C_{\theta}\frac{\p A_{\theta}}{\p \theta}x_{t-1}+\frac{\p  \sigma^{\varepsilon}_{\theta}}{\p \theta}\varepsilon_t+ C_{\theta}\frac{\p u_{t}}{ \p \theta}(\theta)+C_{\theta}\frac{\p  \sigma^{\eta}_{\theta}}{\p \theta}\eta_t\right.\nonumber\\
&&+\left.\frac{\p C_{\theta}}{\p \theta}u_t(\theta)+\frac{\p C_{\theta}}{\p \theta}A_{\theta}x_{t-1}+\frac{\p C_{\theta}}{\p \theta} \sigma^{\eta}_{\theta}\eta_{t}\right)\nonumber\\
&&+\epsilon^{2}\left(\frac{\p C_{\theta}}{\p \theta}\frac{\p u_t}{\p \theta}(\theta)+\frac{\p C_{\theta}}{\p \theta}\frac{\p A_{\theta}}{\p \theta}x_{t-1}+\frac{\p C_{\theta}}{\p \theta}\frac{\p  \sigma^{\eta}_{\theta}}{\p \theta}\eta_{t}\right)+o(\epsilon) \label{erreur_ekf2}
\end{eqnarray}
Define,

\begin{eqnarray*}
&&\mathcal{E}_{y^{-}}^{\epsilon}(\theta)=-\epsilon \left(\frac{\p d_t}{\p \theta}(\theta)+ C_{\theta}\frac{\p u_t}{\p \theta}(\theta)+\frac{\p C_{\theta}}{\p \theta}u_t(\theta)\right),\\
&&\mathcal{F}_{y^{-}}^{\epsilon}(\theta)=-\epsilon \left(C_{\theta}\frac{\p A_{\theta}}{\p \theta}+\frac{\p C_{\theta}}{\p \theta}A_{\theta}\right),\\
&&\mathcal{W}_{y^{-}}^{\epsilon}(\theta)=-\epsilon \left( \frac{\p C_{\theta}}{\p \theta} \sigma^{\eta}_{\theta}\eta_t+C_{\theta} \frac{\p  \sigma^{\eta}_{\theta}}{\p \theta}\eta_t+\frac{\p  \sigma^{\varepsilon}_{\theta}}{\p \theta}\varepsilon_t\right),
\end{eqnarray*}
we obtain:

\begin{eqnarray*}
\xi_t^{-}&=&y_t-\E_{\theta}[y_t\vert y_{1:t-1}]\nonumber\\
&=&C_{\theta} A_{\theta}e_{t-1}+  \sigma^{\varepsilon}_{\theta}\varepsilon_t+C_{\theta} \sigma^{\eta}_{\theta}\eta_t+\mathcal{E}_{y^{-}}^{\epsilon}(\theta)+\mathcal{F}_{y^{-}}^{\epsilon}(\theta)x_{t-1}+\mathcal{W}_{y^{-}}^{\epsilon}(\theta)+o(\epsilon)
\end{eqnarray*}
By combining Eq.(\ref{erreur_ekf1}) and Eq.(\ref{erreur_ekf2}), we have:

\begin{eqnarray*} 
e_t&=&x_{t}-\E_{\theta}[x_t \vert y_{1:t}]\nonumber\\
&=&(I_{n_x\times n_x}-K_t C_{\theta})A_{\theta}e_{t-1}-K_t  \sigma^{\varepsilon}_{\theta}\varepsilon_t-K_t C_{\theta} \sigma^{\eta}_{\theta}\eta_t+ \sigma^{\eta}_{\theta}\eta_t+\mathcal{E}_{x}^{\epsilon}(\theta)+\mathcal{F}_{x}^{\epsilon}(\theta)x_{t-1}+\mathcal{W}_{x}^{\epsilon}(\theta)+o(\epsilon)
\end{eqnarray*} 
where,

\begin{eqnarray*} 
&&\mathcal{E}_{x}^{\epsilon}(\theta)=-\epsilon \left( (I_{n_x\times n_x}-K_t C_{\theta})\frac{\p u_t}{\p \theta}(\theta)-K_t\frac{\p d_t}{\p \theta}(\theta)-\frac{\p C_{\theta}}{\p \theta}u_t(\theta)\right),\\
&&\mathcal{F}_{x}^{\epsilon}(\theta)=-\epsilon \left( (I_{n_x\times n_x}-K_t C_{\theta})\frac{\p A_{\theta}}{\p \theta}-\frac{\p C_{\theta}}{\p \theta}A_{\theta}\right),\\
&&\mathcal{W}_{x}^{\epsilon}(\theta)=-\epsilon \left( \frac{\p  \sigma^{\eta}_{\theta}}{\p \theta}\eta_t-K_t C_{\theta} \frac{\p  \sigma^{\eta}_{\theta}}{\p \theta}\eta_t-K_t  \sigma^{\eta}_{\theta} \frac{\p C_{\theta_0}}{\p \theta}\eta_t-K_t \frac{\p  \sigma^{\varepsilon}_{\theta_0}}{\p \theta}\varepsilon_t\right),
\end{eqnarray*}
One can deduce the \emph{Propagation of the residues a posteriori:} 

\begin{eqnarray*}
\xi_t&=&y_t-\E_{\theta}[y_t\vert y_{1:t}]\\
&=&d_t(\theta_0)+C_{\theta_0}x_t+ \sigma^{\varepsilon}_{\theta_0}\varepsilon_t-\E_{\theta}[d_t(\theta)+C_{\theta}x_t+ \sigma^{\varepsilon}_{\theta}\varepsilon_t\vert y_{1:t}]\\
&=&d_t(\theta_0)-d_t(\theta)+(C_{\theta}-\epsilon \frac{\p C_{\theta}}{\p \theta})x_t-C_{\theta}\E_{\theta}[x_t\vert y_{1:t}]+( \sigma^{\varepsilon}_{\theta}-\epsilon \frac{\p  \sigma^{\varepsilon}_{\theta}}{\p \theta})\varepsilon_t+o(\epsilon)\\
&=&C_{\theta} e_{t}+  \sigma^{\varepsilon}_{\theta}\varepsilon_t-\epsilon \left(\frac{\p d_t}{\p \theta}(\theta)+\frac{\p C_{\theta}}{\p \theta}x_t+\frac{\p  \sigma^{\varepsilon}_{\theta}}{\p \theta}\varepsilon_t\right)+o(\epsilon)
\end{eqnarray*}

By defining:

\begin{eqnarray*} 
&&\mathcal{E}_{y}^{\epsilon}(\theta)=-\epsilon\frac{\p d_t(\theta)}{\p \theta}\\
&&\mathcal{F}_{y}^{\epsilon}(\theta)=-\epsilon \frac{\p C_{\theta}}{\p \theta}\\
&&\mathcal{W}_{y}^{\epsilon}(\theta)=-\epsilon\frac{\p  \sigma^{\varepsilon}_{\theta}}{\p \theta}\varepsilon_t
\end{eqnarray*}Eq.(\ref{f2}) follows. \qed\\

The proof of Corollary \ref{coroll_estierorr1} is obtained by taking the expectations in Eq.(\ref{f1}) and Eq.(\ref{f2}).

\section{\label{preuves2} Proof of Theorem \ref{prop}:} By using the system (\ref{linearisemodel}) and Eq.(\ref{f1})-Eq.(\ref{f2}) we can rewrite the model as follows: 

\begin{equation}\label{prop_sate}
E_t = \begin{pmatrix}
			 \mathcal{E}_{x}^{\epsilon}(\theta)\\
			 u_t(\theta_0)
			 \end{pmatrix}
			 +
			 \begin{pmatrix}
			 (I-K_tC_{\theta})A_{\theta} &\mathcal{F}_{x}^{\epsilon}(\theta)\\
			 0 & A_{\theta_0}
			 \end{pmatrix}
			 E_{t-1}+
			 \begin{pmatrix}
			\mathcal{\tilde{W}}_{x}^{\epsilon}(\theta) \\
			  \sigma^{\eta}_{\theta_0}\eta_t
			 \end{pmatrix}
\end{equation}	
and,

\begin{equation}\label{prop_measure}
\mathfrak{E}_t= \begin{pmatrix}
			 \mathcal{E}_{y}^{\epsilon}(\theta)\\
			 d_t(\theta_0)
			 \end{pmatrix}
			 +
			 \begin{pmatrix}
			 C_{\theta}& \mathcal{F}_{y}^{\epsilon}(\theta)\\
			 0 & C_{\theta_0}
			 \end{pmatrix}
			 E_{t}+
			 \begin{pmatrix}
			\mathcal{\tilde{W}}_{y}^{\epsilon}(\theta)\\
			  \sigma^{\varepsilon}_{\theta_0}\varepsilon_t
			 \end{pmatrix}
\end{equation}		
Hence, the variance matrix $\Sigma_t^x$ is given by: 

\begin{equation*}
\Sigma_t^x=\begin{pmatrix}
			 (I-K_tC_{\theta})A_{\theta} &\mathcal{F}_{x}^{\epsilon}(\theta)\\
			 0 & A_{\theta_0}
			 \end{pmatrix}\Sigma_{t-1}^x\begin{pmatrix}
			 (I-K_tC_{\theta})A_{\theta} &\mathcal{F}_{x}^{\epsilon}(\theta)\\
			 0 & A_{\theta_0}
			 \end{pmatrix}^{'}+\begin{pmatrix} \V[\mathcal{\tilde{W}}_{x}^{\epsilon}(\theta)] & \C(\mathcal{\tilde{W}}_{x}^{\epsilon}(\theta), \sigma^{\eta}_{\theta_0}\eta_t)\\
\C(\mathcal{\tilde{W}}_{x}^{\epsilon}(\theta), \sigma^{\eta}_{\theta_0}\eta_t) & \sigma^{\eta}_{\theta_0}	 \sigma^{'\eta}_{\theta_0}
\end{pmatrix}		 
\end{equation*} 
Additionally, the variance matrix $\Sigma_t^y$ is given by:

\begin{equation*}
\Sigma_t^y= \begin{pmatrix}
			 C_{\theta}& \mathcal{F}_{y}^{\epsilon}(\theta)\\
			 0 & C_{\theta_0}
			 \end{pmatrix}\Sigma_t^x\begin{pmatrix}
			 C_{\theta}& \mathcal{F}_{y}^{\epsilon}(\theta)\\
			 0 & C_{\theta_0}
			 \end{pmatrix}^{'}+\begin{pmatrix} \V[\mathcal{\tilde{W}}_{y}^{\epsilon}(\theta)] & \C(\mathcal{\tilde{W}}_{y}^{\epsilon}(\theta), \sigma^{\varepsilon}_{\theta_0}\eta_t)\\
\C(\mathcal{\tilde{W}}_{x}^{\epsilon}(\theta), \sigma^{\varepsilon}_{\theta_0}\eta_t) & \sigma^{\varepsilon}_{\theta_0}	 \sigma^{'\varepsilon}_{\theta_0}
\end{pmatrix}		 
\end{equation*} 
Proposition \ref{prop_ekf} gives that:

 \begin{eqnarray*}
\left\lbrace\begin{array}{ll}
\mathcal{W}_{x}^{\epsilon}(\theta)=-\epsilon \left( \frac{\p  \sigma^{\eta}_{\theta_0}}{\p \theta}\eta_t-K_t C_{\theta} \frac{\p  \sigma^{\eta}_{\theta_0}}{\p \theta}\eta_t-K_t  \sigma^{\eta}_{\theta} \frac{\p C_{\theta_0}}{\p \theta}\eta_t-K_t \frac{\p  \sigma^{\varepsilon}_{\theta_0}}{\p \theta}\varepsilon_t\right)\\
\mathcal{W}_{y}^{\epsilon}(\theta)=-\epsilon\frac{\p  \sigma^{\varepsilon}_{\theta_0}}{\p \theta}\varepsilon_t
\end{array}
\right.
\end{eqnarray*}
Hence, if $\epsilon << 1$, then

\begin{eqnarray*}
\V[\mathcal{W}_{x}^{\epsilon}(\theta)]=Q_{\theta}+K_t\left(C_{\theta}Q_{\theta}C_{\theta}^{'}+R_{\theta}\right)K_{t}^{'} \text{ and }\V[\mathcal{W}_{y}^{\epsilon}(\theta)]=R_{\theta}
\end{eqnarray*}
Furthermore, the covariances are given by:

\begin{eqnarray*}
Cov_{\theta}\left( \mathcal{\tilde{W}}_{x}^{\epsilon}(\theta),\sigma^{\eta}_{\theta_0} \eta_{t} \right)&=&-\epsilon Cov_{\theta}\left(\frac{\p \sigma^{\eta}_{\theta_0} }{\p \theta} \eta_t, \sigma^{\eta}_{\theta_0}  \eta_t\right)+\epsilon\left(K_t C_{\theta}\frac{\p \sigma^{\eta}_{\theta_0} }{\p \theta}\eta_t, \sigma^{\eta}_{\theta_0} \eta_t\right)+\epsilon Cov_{\theta}\left(K_t \sigma^{\eta}_{\theta}\eta_t, \sigma^{\eta}_{\theta_0}\eta_t\right)\\
&+&\epsilon Cov_{\theta}\left(K_t \frac{\p \sigma^{\eta}_{\theta_0} }{\p \theta}\varepsilon_t, \sigma^{\eta}_{\theta_0}\eta_t \right)\\
&=&\epsilon \frac{\p \sigma^{\eta}_{\theta_0} }{\p \theta}Q_{\theta_0} \sigma^{'\eta}_{\theta_0}+\epsilon K_t C_{\theta}\frac{\p \sigma^{\eta}_{\theta_0} }{\p \theta}Q_{\theta_0}\sigma^{'\eta}_{\theta_0}+\epsilon K_t \sigma^{\eta}_{\theta}Q_{\theta_0}\sigma^{'\eta}_{\theta_0} \text{ by assumption \textbf{A2}}\\
&=&-\epsilon\left(\frac{\p \sigma^{\eta}_{\theta_0} }{\p \theta} Q_{\theta_0}\sigma^{'\eta}_{\theta_0}-K_t \left(C_{\theta}\frac{\p \sigma^{\eta}_{\theta_0} }{\p \theta}+\sigma^{\eta}_{\theta}\right)Q_{\theta_0}\sigma^{'\eta}_{\theta_0}\right)
\end{eqnarray*}
Additionally,

\begin{eqnarray*}
Cov_{\theta}\left( \mathcal{\tilde{W}}_{y}^{\epsilon}(\theta),\sigma^{\varepsilon}_{\theta_0} \varepsilon_{t} \right)&=&-\epsilon Cov_{\theta}\left(\frac{\p \sigma^{\varepsilon}_{\theta_0} }{\p \theta} \varepsilon_t, \sigma^{\varepsilon}_{\theta_0}  \varepsilon_t\right)\\
&=& -\epsilon \left(\frac{\p \sigma^{\varepsilon}_{\theta_0} }{\p \theta}R_{\theta_0}\sigma^{'\varepsilon}_{\theta_0}\right)
\end{eqnarray*}
\qed

\bibliographystyle{alpha} 
\bibliography{Biblio}

\begin{acknowledgements}
  The author wishes to thank Fr{\'e}d{\'e}ric Patras for his supervisory throughout this paper, Patricia Reynaud-Bouret and N. Chopin for their suggestions and their interest about this framework.
\end{acknowledgements}

\end{document}